\documentclass{article}
\usepackage{amsmath,amssymb,theorem}

\newcommand{\email}[1]{{\textit{Email:} \texttt{#1}}}
\newcommand{\tmmathbf}[1]{\ensuremath{\boldsymbol{#1}}}
\newcommand{\tmop}[1]{\ensuremath{\operatorname{#1}}}
\newcommand{\tmstrong}[1]{\textbf{#1}}
\newcommand{\tmtextit}[1]{{\itshape{#1}}}
\newenvironment{itemizedot}{\begin{itemize} }{\end{itemize}}
\newenvironment{proof}{\noindent\textbf{Proof\ }}{\hspace*{\fill}$\Box$\medskip}
\newtheorem{corollary}{Corollary}
\newtheorem{definition}[corollary]{Definition}
{\theorembodyfont{\rmfamily}\newtheorem{example}[corollary]{Example}}
\newtheorem{lemma}[corollary]{Lemma}
{\theorembodyfont{\rmfamily}\newtheorem{remark}[corollary]{Remark}}
\newtheorem{theorem}[corollary]{Theorem}

\begin{document}

\title{Renormalization and quantum field theory}\author{R. E.
Borcherds\thanks{This research was supported by a Miller professorship and an
NSF grant. I thank the referees for suggesting many improvements.}\\
Department of Mathematics\\
University of California at Berkeley\\
CA 94720-3840 USA\\ \email{reb@math.berkeley.edu}}\maketitle

\begin{abstract}
  The aim of this paper is to describe how to use regularization and
  renormalization to construct a perturbative quantum field theory from a
  Lagrangian. We first define renormalizations and Feynman measures, and show
  that although there need not exist a canonical Feynman measure, there is a
  canonical orbit of Feynman measures under renormalization. We then construct
  a perturbative quantum field theory from a Lagrangian and a Feynman measure,
  and show that it satisfies perturbative analogues of the Wightman axioms,
  extended to allow time-ordered composite operators over curved spacetimes.
\end{abstract}

\section{Introduction}

We give an overview of the construction of a perturbative quantum field theory
from a Lagrangian. We start by translating some terms in physics into
mathematical terminology.

\begin{definition}
  Spacetime is a smooth finite-dimensional metrizable manifold $M$, together
  with a ``causality'' relation $\leqslant$ that is closed, reflexive, and
  transitive. We say that two points are {\tmstrong{spacelike separated}} if
  they are not comparable, in other words neither $x \leqslant y$ nor $y
  \leqslant x$. \ 
\end{definition}

The causality relation $a \leqslant b$ means informally that $a$ occurs before
$b$. The causality relation will often be constructed in the usual way from a
Lorentz metric with a time orientation, but since we do not use the Lorentz
metric for anything else we do not bother to give $M$ one. The Lorentz metric
will later appear implicitly in the choice of a cut propagator, which is often
constructed using a metric. {\tmstrong{}}

\begin{definition}
  The sheaf of classical fields $\Phi$ is the sheaf of smooth sections of some
  finite dimensional super vector bundle over spacetime. 
\end{definition}

When the sheaf of classical fields is a ``super-sheaf'', one uses the usual
conventions of superalgebra: in particular the symmetric algebras used later
are understood to be symmetric algebras in the superalgebra sense, and the
usual superalgebra minus signs should be inserted into formulas whenever the
order of two terms is exchanged.

As usual, a global section of a sheaf of things is called a thing, so a
classical field $\varphi$ is a global section of the sheaf $\Phi$ of classical
fields, and so on. (A subtle point is sometimes things called classical fields
in the physics literature are better thought of as sections of the
{\tmstrong{dual}} of the sheaf of classical fields; in practice this
distinction does not matter because the sheaf of classical fields usually
comes with a bilinear form giving a canonical isomorphism with its dual.)

\begin{definition}
  The sheaf of derivatives of classical fields or simple fields is the sheaf
  $J \Phi = \tmop{Hom} (J, \Phi)$, where $J$ is the sheaf of jets of $M$ and
  the $\tmop{Hom}$ is taken over the smooth functions on $M$, equal to the \
  inverse limit of the sheaves of jets of finite order of $M$, as in
  {\cite[16.3]{Grothendieck}}.
\end{definition}

\begin{definition}
  The sheaf of (polynomial) Lagrangians or composite fields $SJ \Phi$
  is the symmetric algebra of the sheaf $J \Phi$ of derivatives of classical
  fields.
\end{definition}

Its sections are (polynomial) Lagrangians, in other words polynomial in
fields and their derivations, so for example $\lambda \varphi^4 + m^2
\varphi^2 + \varphi \partial_i^2 \varphi$ is a Lagrangian, but $\sin
(\varphi)$ is not.

Perturbative quantum field theories depend on the choice of a Lagrangian $L$,
which is the sum of a free Lagrangian $L_F$ that is quadratic in the fields,
and an interaction Lagrangian $L_I \in SJ \Phi \otimes
\tmmathbf{C}[[\lambda_1, \ldots, \lambda_n]]$ whose coefficients are
infinitesimal, in other words elements of a formal power series ring
$\tmmathbf{C}[[\lambda_1, \ldots, \lambda_n]]$ over the reals with constant
terms 0.

\begin{definition}
  The sheaf of Lagrangian densities or local actions \ $\omega SJ \Phi
  = \omega \otimes SJ \Phi$ is the tensor product of the sheaf
  $SJ \Phi$ of Lagrangians and the sheaf $\omega$ of smooth densities
  (taken over smooth functions on $M$).
\end{definition}

For a smooth manifold, the (dualizing) sheaf $\omega$ of smooth densities (or
smooth measures) is the tensor product of the orientation sheaf with the sheaf
of differential forms of highest degree, and is non-canonically isomorphic to
the sheaf of smooth functions. Densities are roughly ``things that can be
locally integrated''. For example, if $M$ is oriented, then \ $(\lambda
\varphi^4 + m^2 \varphi^2 + \varphi \partial_i^2 \varphi) d^n x$ is a
Lagrangian density.

We use $\Gamma$ and $\Gamma_c$ to stand for spaces of global and compactly
supported sections of a sheaf. These will \ usually be spaces of smooth
functions (or compactly supported smooth functions) in which case they are
topologized in the usual way so that their duals are compactly supported
distributions (or distributions) taking values in some sheaf.

\begin{definition}
  A (non-local) action is a polynomial in local actions, in other words an
  element of the symmetric algebra $S \Gamma \omega SJ \Phi$ of the
  real vector space $\Gamma \omega SJ \Phi$ of \ local actions.
\end{definition}

We do not complete the symmetric algebra, so expressions such as $e^{i \lambda
L}$ are not in general non-local actions, unless we work over some base ring
in which $\lambda$ is nilpotent.

We will use $\ast$ for complex conjugation and for the antipode of a Hopf
algebra and for the adjoint of an operator and for the anti-involution of a
$*$-algebra. The use of the same symbol for all of these is deliberate and  indicates that they are all really special cases of a universal ``adjoint'' or ``antipode'' operation that acts on everything: whenever two of these operations are defined on something they are equal, so can all be denoted by the same symbol. 

The quantum field theories we construct depend on the choice of a \ cut
propagator $\Delta$ that is essentially the same as the 2-point Wightman
distribution
\[ \Delta (\varphi_1, \varphi_2) = \int_{x, y} \langle 0| \varphi_1 (x)
   \varphi_2 (y) |0 \rangle dx dy \]
\begin{definition}
  A propagator $\Delta$ is a continuous bilinear map $\Gamma_c \omega \Phi
  \times \Gamma_c \omega \Phi \rightarrow \tmmathbf{C}$.
  \begin{itemizedot}
    \item $\Delta$ is called {\tmstrong{local}} if $\Delta (f, g) = \Delta (g,
    f)$ whenever the supports of $f$ and $g$ are spacelike separated.
    
    \item $\Delta$ is called {\tmstrong{Feynman}} if it is symmetric: $\Delta
    (f, g) = \Delta (g, f)$.
    
    \item $\Delta$ is called {\tmstrong{Hermitian}} if $\Delta^{\ast} =
    \Delta$, where $\Delta^{\ast}$ is defined by $\Delta^{\ast} (f^{\ast},
    g^{\ast}) = \Delta (g, f)^{\ast}$ (with a change in order of $f$ and $g$).
    
    \item $\Delta$ is called {\tmstrong{positive}} if $\Delta (f^{\ast}, f)
    \geqslant 0$ for all $f$.
    
    \item $\Delta$ is called \ {\tmstrong{cut}} if it satisfies the following
    ``positive energy'' condition: at each point $x$ of $M$ there is a partial
    order on the cotangent space defined by a proper closed convex cone $C_x$,
    \ such that if $(p, q)$ is in the wave front set of $\Delta$ at some point
    $(x, y) \in M^2$ \ then $p \leqslant 0$ and $q \geqslant 0$. Also, as a
    distribution, $\Delta$ can be written in local coordinates as a boundary
    value of something in the algebra generated by smooth functions and powers
    and logarithms of polynomials (the boundary values taken so that the wave
    front sets lie in the regions specified above). Moreover if $x = y$ then
    $p + q = 0$.
  \end{itemizedot}
\end{definition}

A propagator can also be thought of as \ a complex distribution on $M
\times M$ taking values in the dual of the external tensor product $J
\Phi \boxtimes J \Phi$. In particular it has a wave front set (see
H\"ormander {\cite{Hormander}}) at each point of $M^2$, which is a
cone in the imaginary cotangent space of that point. If $A$ and $B$
are in $\Gamma_c \Phi$, then $\Delta (A, B)$ is defined to be a
compactly supported distribution on $M \times M$, defined by \ $\Delta
(A, B) (f, g) = \text{$\Delta (Af, Bf)$}$ for $f$ and $g$ in $\Gamma
\omega$.

The key point in the definition of a cut propagator is the condition on the
wave front sets, which distinguishes the cut propagators from other
propagators such as Feynman propagators or advanced and retarded propagators
that can have more complicated wave front sets. For most common cut
propagators in Minkowski space, this follows from the fact that their Fourier
transforms have support in the positive cone. The condition about being
expressible in terms of smooth functions and powers and logs of polynomials is
a minor technical condition that is in practice satisfied by almost any
reasonable example, and is used in the proof that Feynman measures exist.

If $(p_1, \ldots p_n)$ is in the imaginary cotangent space of a point of
$M^n$, then we write $(p_1, \ldots p_n) \geqslant 0$ if $p_j \geqslant 0$ for
all $j$, and call it positive if it is not zero. \

\begin{example}
  Over Minkowski space, most of the usual cut propagators are positive (except
  for ghost fields), local, and Hermitian. Most of the ideas for the proof of
  this can be seen for the simplest case of the propagator for massive
  Hermitian scalar fields. Using translation invariance, we can write $\Delta
  (x, y) = \Delta (x - y)$ for some distribution $\Delta$ on Minkowski
  spacetime. Then the Fourier transform of this in momentum space is a
  rotationally invariant measure supported on one of the two components of
  vectors with $p^2 = m^2$. This propagator is positive because the measure in
  momentum space is positive. It satisfies the wave front set part of the cut
  condition because the Fourier transform has support in the positive cone,
  and explicit calculation shows that it can be written in terms of powers and
  logs of polynomials. It satisfies locality because it is invariant under
  rotations that preserve the direction of time, and under such rotations any
  space-like vector is conjugate to its negative, so $\Delta (x) = \Delta (-
  x)$ whenever $x$ is spacelike, in other words $\Delta (x, y) = \Delta (y,
  x)$ whenever $x$ and $y$ are spacelike separated. The corresponding Feynman
  propagator is given by $1 / (p^2 + m^2 + i \varepsilon)$ where the $i
  \varepsilon$ indicates in which direction one integrates around the poles,
  so the cut propagator is just the residue of the Feynman propagator along
  one of the 2 components of the 2-sheeted hyperboloid \ $p^2 = m^2$.
  
  For other fields such as spinor fields in Minkowski space, the sheaf of
  classical fields will usually be some sort of spin bundle. The propagators
  can often be expressed in terms of the \ the propagator for a scalar field
  by acting on it with polynomials in momentum multiplied by Dirac's gamma
  matrices $\gamma^{\mu}$, for example $i (\gamma^{\mu} p_{\mu} + m) / (p^2 -
  m^2)$. Unfortunately there are a bewildering number of different notational
  and sign conventions for gamma matrices.
\end{example}

Compactly supported actions give functions on the space $\Gamma \Phi$ of
smooth fields, by integrating over spacetime $M$. \ A Feynman measure is a
sort of analogue of Haar measure on a finite dimensional real vector space. We
can think of a Haar measure as an element of the dual of the space of
continuous compactly supported functions. For infinite dimensional vector
spaces there are usually not enough continuous compactly supported functions,
but instead we can define a measure to be an element of the dual of some other
space of functions. We will think of Feynman measures as something like
elements of the dual of all functions that are given by free field Gaussians
times a compactly supported action. In other words a Feynman measure should
assign a complex number to each compactly supported action, formally
representing the integral over all fields of this action times a Gaussian
$e^{iL_F}$, where we think of the action as a function of classical fields (or
rather sections of the dual of the space of classical fields, which can
usually be identified with classical fields). Moreover the Feynman measure
should satisfy some sort of analogue of translation invariance.

The space $e^{iL_F} S \Gamma_c \omega SJ \Phi$ is a free rank 1 module over $S
\Gamma_c \omega SJ \Phi$ generated by the basis element $e^{iL_F}$, which can
be thought of either as a formal symbol or a formal power series. Its elements
can be thought of as representing functions of classical fields that are given
by a polynomial times the Gaussian \ $e^{iL_F}$, and will be the functions
that the Feynman measure is defined on. The symmetric algebra $S \Gamma_c
\omega SJ \Phi$ is topologized as the direct sum of the spaces $S^n \Gamma_c
\omega SJ \Phi$, each of which is toplogized by regarding it as a space of
smooth test functions over $M^n$.

For the definition of a Feynman measure we need to  extend the propagator
$\Delta$ to a larger space as follows.  We think of the propagator $\Delta$ as a
map taking $\Gamma_cJ\Phi\otimes \Gamma_cJ\Phi$ to distributions on
$M\times M$.  We then extend it a map from $\Gamma_c S J\Phi \times
\Gamma_c S J \Phi $ to distributions on $M\times M$ by putting
$\Delta(a_1\cdots a_n,b_1\cdots b_n) =\sum_{\sigma\in S_n}
\Delta(a_1,b_{\sigma(1)})\times\cdots\times\Delta(a_1,b_{\sigma(n)})$
where the sum is over all elements of the symmetric group $S_n$ (and
defining it to be 0 for arguments of different degrees). Finally we
extend it to a map from $S^m\Gamma_c S J\Phi \times S^n\Gamma_c S J
\Phi $ to distributions on $M^m\times M^n$ using the ``bicharacter'' property: in other words $\Delta(AB,C) = \sum\Delta(A,C')\Delta(B,C'')$ where the coproduct of $C$ is $\sum C'\otimes C''$, and similarly for $\Delta(A,BC)$. 

\begin{definition}
  \label{feynman measure}A Feynman measure is a continuous linear map $\omega
  : e^{iL_F} S \Gamma_c \omega SJ \Phi \rightarrow \tmmathbf{C}$. The Feynman
  measure is said to be associated with the propagator $\Delta$ if it
  satisfies the following conditions:
  \begin{itemize}
    \item Smoothness on the diagonal: Whenever $(p_1, \ldots, p_n)$ is in the
    wave front set of $\omega$ at the point $(x, \ldots, x)$ on the diagonal,
    then $p_1 + \ldots + p_n = 0$
    
    \item Non-degeneracy: there is a smooth nowhere-vanishing function $g$ so
    that $\omega (e^{iL_F} v)$ is \ $\int_M gv$ for $v$ in $\Gamma_c \omega
    S^0 J \Phi = \Gamma_c \omega$.
    
    \item Gaussian condition, or weak translation invariance: For \ $A\in S^m\Gamma_c \omega SJ \Phi$, $B\in S^n\Gamma_c \omega SJ \Phi$,
    with both sides interpreted as distributions on $M^{m+n}$,
    \[ \text{$\omega (AB) = \sum \omega (A') \Delta (A_{}'', B'') \omega
       (B')$} \]
    whenever there is no element in the support of $A_{}$ that is $\leqslant$
    some element of the support of \ $B_{}$. Here $\sum A' \otimes A'' \in S
    \Gamma_c \omega SJ \Phi \otimes S\Gamma_c\text{$SJ \Phi$}$ is the
    image of $A$ under the map $S^m \Gamma_c \omega SJ \Phi
    \rightarrow S^m \Gamma_c \omega SJ \Phi \otimes S^m\Gamma_c\text{$SJ
    \Phi$}$ induced by the coaction $ \omega SJ \Phi
    \rightarrow  \omega SJ \Phi \otimes \text{$SJ
    \Phi$}$ of $SJ \Phi$ on $ \omega SJ \Phi$, and
    similarly for $B$. The product on the right is  a product of
    distributions, using the extended version of $\Delta$ defined just before this definition. 
  \end{itemize}
\end{definition}

We explain what is going on in this definition. We would like to
define the value of the Feynman measure to be a sum over Feynman
diagrams, formed by joining up pairs of fields in all possible ways by
lines, and then assigning a propagator to each line and taking the
product of all propagators of a diagram. This does not work because of
ultraviolet divergences: products of propagators need not be defined
when points coincide. If these products were defined then they would
satisfy the Gaussian condition, which then says roughly that if the
vertices are divided into two disjoint subsets $a$ and $b$, then a
Feynman diagram can be divided into a subdiagram with vertices $a$, a
subdiagram with vertices $b$, and some lines between $a$ and $b$. The
value $\omega(AB)$ of the Feynman diagram would then be the product of
its value $\omega(A')$ on $a$, the product $\Delta(A'',B'')$ of all
the propagators of lines joining $a$ and $b$, and its value
$\omega(B')$ on $b$. The Gaussian condition need not make sense if
some point of $a$ is equal to some point of $b$ because if these
points are joined by a line then the corresponding propator may have a
bad singularity, but does make sense whenever all points of $a$ are
not $\le$ all points of $b$. The definition above says that a Feynman
measure should at least satisfy the Gaussian condition in this case,
when the product is well defined.

Unfortunately the standard notation $\omega$ for a dualizing sheaf, such as
the sheaf of densities, is the same as the standard notation $\omega$ for a
state in the theory of operator algebras, which the Feynman measure will be a
special case of. It should be clear from the context which meaning of $\omega$
is intended. \

If $\omega$ is a Feynman measure and $A \in e^{iL_F} S^n \Gamma_c \omega
SJ \Phi$ then $\omega (A)$ is a complex number, but can also be
considered as the compactly supported density on $M^n$ taking a smooth $f$ to
$\omega (A) (f) = \omega (Af)$. The integral of this density $\omega (A)$ over
spacetime is just the complex number $\omega (A)$.

Since $e^{iL_F} S \Gamma_c \omega SJ \Phi$ is a coalgebra (where
elements of $\Gamma_c \omega SJ \Phi$ are primitive and $e^{iL_F}$ is
group-like), the space of Feynman measures is an algebra, whose product is
called convolution.

The non-degeneracy condition just excludes some uninteresting degenerate
cases, such as the measure that is identically zero, and the function $g$
appearing in it is usually normalized to be 1. The condition about smoothness
on the diagonal implies that the product on the right in the Gaussian
condition is defined. This is because $\omega$ has the property that if an
element $(p_1, \ldots p_n)$ of the wave front set of some point is nonzero
then its components cannot all be positive and cannot all be negative. This
shows that the wave front sets are such that the product of distributions is
defined.

If $A$ is in $e^{iL_F} S \Gamma_c \omega SJ \Phi$, then $\omega (A)$
can be thought of as a Feynman integral
\[ \omega (A) = \int A (\varphi)\mathcal{D} \varphi \]
where $L_F$ is a quadratic action with cut propagator $\Delta$, and where $A$
is considered to be a function of fields $\varphi$. The integral is formally
an integral over all classical fields. The Gaussian condition is a weak form
of translation invariance of this measure under addition of classical fields.
Formally, translation invariance is equivalent to the Gaussian condition with
the condition about supports omitted and cut propagators replaced by Feynman
propagators, but this is not well defined because the Feynman propagators can
have such bad singularities that their products are sometimes not defined when
two spacetime points coincide.

The Feynman propagator $\Delta_F$ of a Feynman measure $\omega$ is defined to
be the restriction of $\omega$ to $\Gamma_c \omega \Phi \times \Gamma_c \omega
\Phi$. It is equal to the cut propagator at ``time-ordered'' points $(x, y)
\in M^2$ where $x \nleqslant y$, but will usually differ if $x \leqslant y$.
As it is symmetric, it is determined by the cut propagator except on the
diagonal of $M \times M$. Unlike cut propagators, Feynman propagators may have
singularities on the diagonal whose wave front sets are not contained in a
proper cone, so that their products need not be defined.

Any symmetric algebra $\tmop{SX}$ over a module $X$ has a natural structure
of a commutative and cocommutative Hopf algebra, with the coproduct defined by
making all elements of $X$ primitive (in other words, $\Delta x = x \otimes 1
+ 1 \otimes x$ for $x \in X$). In other words, $\tmop{SX}$ is the coordinate
ring of a commutative affine group scheme whose points form the dual of $X$
under addition. \ For general results about Hopf algebra see Abe {\cite{Abe}}.
Similarly $SJ \Phi$ is a sheaf of commutative cocommutative Hopf
algebras, with a coaction on itself and the trivial coaction on $\omega$, and
so has a coaction on $S \omega SJ \Phi$, preserving the coproduct of $S
\omega SJ \Phi$. It corresponds to the sheaf of commutative affine
algebraic groups whose points correspond to the sheaf $J \Phi$ under addition.

\begin{definition}
  A renormalization is an \ automorphism of $S \omega SJ \Phi$
  preserving its coproduct and the coaction of $SJ \Phi$. The group of
  renormalizations is called the ultraviolet group.
\end{definition}

\ The justification for this rather mysterious definition is theorem
\ref{transitive}, which shows that renormalizations act simply transitively on
the Feynman measures associated to a given local cut propagator. In other
words, although there is no canonical Feynman measure on the space of
classical fields, there is a canonical orbit of such measures under
renormalization.

More generally, renormalizations are global sections of the sheaf of
renormalizations (defined in the obvious way), but we will make no use of this
point of view.

The (infinite dimensional) ultraviolet group really ought to be called
the ``renormalization group'', but unfortunately this name is already
used for a quite different 1-dimensional group.  The
``renormalization group'' is the group of positive real numbers,
together with an action on Lagrangians by ``renormalization group flow''.
The relation between the renormalization group and the ultraviolet
group is that the renormalization group flow can be thought of as a
non-abelian 1-cocycle of the renormalization group with values in the
ultraviolet group, using the action of renormalizations on Lagrangians
that will be constructed later.

The ultraviolet group is indirectly related to the Hopf algebras of Feynman
diagrams introduced by Kreimer {\cite{Kreimer}} and applied to renormalization
by him and Connes {\cite{Connes}}, though this relation is not that easy to
describe. First of all their Hopf algebras correspond to Lie algebras, and the
ultraviolet group has a Lie algebra, and these two Lie algebras are related.
There is no direct relation between Connes and Kreimer's Lie agebras and the
Lie algebra of the ultraviolet group, in the sense that there seems to be no
natural homomorphism in either direction. However there seems to be a sort of
intermediate Lie algebra that has homomorphisms to both. This intermediate Lie
algebra (or group) can be defined using Feynman diagrams decorated with smooth
test functions rather than the sheaf $S \omega SJ \Phi$ used here.
Unfortunately all my attempts to explain the product of this Lie algebra
explicitly have resulted in an almost incomprehensible combinatorial mess so
complicated that it is unusable. Roughly speaking, the main differences
between the ultraviolet group and the intermediate Lie algebra is that the Lie
algebra of the ultraviolet group amalgamates all Feynman diagrams with the
same vertices while the intermediate Lie algebra algebra keeps track of
individual Feynman diagrams, and the main difference between the intermediate
Lie algebra and Kreimer's algebra is that \ the intermediate Lie algebra is
much fatter than Kreimer's algebra because it has infinite dimensioinal spaces
of smooth functions in it. In some sense Kreimer's algebra could be thought of
as a sort of skeleton of the intermediate Lie algebra. \

All reasonable Feynman measures for a given free field theory are equivalent
up to renormalization, but it is not easy to show that at least one exists. We
do this by following the usual method of constructing a perturbative quantum
field theory in physics. We \ first regularize the cut local propagator which
produces a meromorphic family of Feynman measures, following Etingof
{\cite{Etingof}} in using Bernstein's theorem {\cite{Bernstein}} on the
analytic continuation of powers of a polynomial to construct the
regularization. We then use an infinite renormalization to eliminate the poles
of the regularized Feynman measure in order of their complexity.

A quantum field theory satisfying the Wightman axioms {\cite[section
3.1]{Streater}} is determined by its Wightman distributions, which are given
by linear maps $\omega_n : T^n \Gamma_c \omega \Phi \rightarrow \tmmathbf{C}$
from the tensor powers of the space of test functions for each $n$. We will
follow H. J. Borchers {\cite{Borchers}} in combining the Wightman
distributions into a Wightman functional $\omega : \text{$T \Gamma_c \omega
\Phi$} \rightarrow \tmmathbf{C}$ on the tensor algebra $T \Gamma_c \omega
\Phi$ of the space $\Gamma_c \omega \Phi$ of test functions (which is
sometimes called a Borchers algebra or Borchers-Uhlmann algebra or
BU-algebra). In order to accommodate composite operators we extend the algebra
$T \Gamma_c \omega \Phi$ to the larger algebra $T \Gamma_c \omega SJ
\Phi$, and to accommodate time ordered operators we extend it further to
$\tmop{TS} \Gamma_c \omega SJ \Phi$. In this set up it is clear how to
accommodate perturbative quantum field theories: we just allow $\omega$ to
take values in a space of formal power series
$\tmmathbf{C}[[\tmmathbf{\lambda}]] =\tmmathbf{C}[[\lambda_1, \lambda_2,
\ldots]]$ rather than $\tmmathbf{C}$. \ For regularization $\omega$ sometimes
takes values in a ring of meromorphic functions. There is one additional
change we need: it turns out that the elements of $\Gamma_c \omega SJ
\Phi$ do not really represent operators on a space of physical states, but are
better thought of as operators that map a space of incoming states to a space
of outgoing states, and vice versa. If we identify the space of incoming
states with the space of physical states, this means that only products of an
even number of operators of $S \Gamma_c \omega SJ \Phi$ act on the
space of physical states. So the functional defining a quantum field theory is
really defined on the subalgebra $T_0 S \Gamma_c \omega SJ \Phi$ of
even degree elements.

So the main goal of this paper is to construct a linear map from $T_0 S
\Gamma_c \omega SJ \Phi$ to $\tmmathbf{C}[\tmmathbf{\lambda}]$ from a
given Lagrangian, and to check that it satisfies analogues of the Wightman
axioms.

The space of physical states of the quantum field theory can be reconstructed
from $\omega$ as follows.

\begin{definition}
  Let \ $\omega : T \rightarrow C$ be a $\tmmathbf{R}$-linear map between real
  $*$-algebras.
  \begin{itemizedot}
    \item $\omega$ is called Hermitian if $\omega^{\ast} = \omega$, where
    $\omega^{\ast} (a^{\ast}) = \omega (a)^{\ast}$
    
    \item $\omega$ is called positive if it maps positive elements to positive
    elements, where an element of a $*$-algebra is called positive if it is a
    finite sum of elements of the form $a^{\ast} a$.
    
    \item $\omega$ is called a state if it is positive and normalized by
    $\omega (1) = 1$
    
    \item The left, right, or 2-sided kernel of $\omega$ is the largest left,
    right or 2-sided ideal closed under * on which $\omega$ vanishes.
    
    \item The space of physical states of $\omega$ is the quotient of $T$ by
    the left kernel of $\omega$. Its sesquilinear form is $\left\langle a, b
    \right\rangle = \omega (a^{\ast} b)$, and its vacuum vector is the image
    of $1$.
    
    \item The algebra of physical operators of $\omega$ is the quotient of $T$
    by the 2-sided kernel of $\omega$.
  \end{itemizedot}
\end{definition}

The algebra of physical operators is a $*$-algebra of operators with a left
action on the physical states. If $\omega$ is positive or Hermitian then so is
the sesquilinear form $\left\langle, \right\rangle$. When $\omega$ is
Hermitian and positive and $C$ is the complex numbers the left kernel of
$\omega$ is the set of vectors $a$ with $\omega (a^{\ast} a) = 0$, and the
definition of the space of physical states is essentially the GNS construction
and is also the main step of the Wightman reconstruction theorem. In this case
the completion of the space of physical states is a Hilbert space.

The maps $\omega$ we construct are defined on the real vector space $T_0 S
\Gamma_c \omega SJ \Phi$ and will initially be $\tmmathbf{R}$-linear.
It is often convenient to extend them to be
$\tmmathbf{C}[[\tmmathbf{\lambda}]]$-linear maps defined on $T_0 S \Gamma_c
\omega SJ \Phi \otimes \tmmathbf{C}[[\tmmathbf{\lambda}]]$, in which
case the corresponding space of physical states will be a module over
$\tmmathbf{C}[[\tmmathbf{\lambda}]]$ and its bilinear form will be
sesquilinear over $\tmmathbf{C}[[\tmmathbf{\lambda}]]$.

The machinery of renormalization and regularization has little to do with
perturbation theory or the choice of Lagrangian: instead, it is needed even
for the construction of free field theories if we want to include composite
operators. The payoff for all the extra work needed to construct the composite
operators in a free field theory comes when we construct interacting field
theories from free ones. The idea for constructing an interacting field theory
from a free one is simple: we just apply a suitable automorphism (or
endomorphism) of the algebra $T_0 S \Gamma_c \omega SJ \Phi$ to the
free field state $\omega$ to get a state for an interacting field. For
example, if we apply an endomorphsim of the sheaf $\omega SJ \Phi$ then
we get the usual field theories of normal ordered products of operators, which
are not regarded as all that interesting. For any Lagrangian $L$ there is an
infinitesimal automorphism of $T_0 S \Gamma_c \omega SJ \Phi$ that just
multiplies elements of $S \Gamma_c \omega SJ \Phi$ by $iL$, which we
would like to lift to an automorphism $e^{iL}$. The construction of an
interacting quantum field theory from a Feynman measure $\omega$ and a
Lagrangian $L$ is then given by the natural action $e^{- iL} \omega$ of the
automorphism $e^{- iL}$ on the state $\omega$. The problem is that $e^{iL_I}$
is only defined if the interaction Lagrangian has infinitesimal coefficients,
due to the fact that we only defined $\omega$ on polynomials times a Gaussian,
so this construction only produces perturbative quantum field theories taking
values in rings of formal power series. \ This is essentially the problem of
lifting a Lie algebra elements $L_I$ to a group element $e^{iL_I}$, which is
trivial for operators on finite dimensional vector spaces, but a subtle and
hard problem for unbounded operators such as $L_I$ that are not self adjoint.
This construction works provided the interacting part of the Lagrangian not
only has infinitesimal coefficients but also has compact support. We show that
the more general case of Lagrangians without compact support can be reduced to
the case of compact support up to inner automorphisms, at least on globally
hyperbolic spacetimes, by showing that infra-red divergences cancel.

Up to isomorphism, the quantum field theory does not depend on the choice of
Feynman measure or Lagrangian, but only on the choice of propagator. In
particular, the interacting quantum field theory is isomorphic to a free one.
This does not mean that interacting quantum field theories are trivial,
because this isomorphism does not preserve the subspace of simple operators,
so if one only looks at the restriction to simple operators, as in the
Wightman axioms, one no longer gets an isomorphism between free and
interacting theories. The difference between interacting and free field
theories is that one chooses a different set of operators to be the ``simple''
operators corresponding to physical fields.

The ultraviolet group also has a non-linear action on the space of
infinitesimal Lagrangians. A quantum field theory is determined by the choice
of a Lagrangian and a Feynman measure, and this quantum field theory is
unchanged if the Feynman measure and the Lagrangian are acted on by the same
renormalization. This shows why the choice of Feynman measure is not that
important: if one chooses a different Feynman measure, it is the image of the
first by a unique renormalization, and by applying this renormalization to the
Lagrangian one still gets the same quantum field theory.

\ Roughly speaking, we show that these quantum field theories $e^{iL_I}
\omega$ satisfy the obvious generalizations of Wightman axioms whenever it is
reasonable to expect them to do so. For example, we will show that locality
holds by showing that the state vanishes on the ``locality ideal'' of
definition \ref{localityideal}, the quantum field theory is Hermitian if we
start with Hermitian cut propagators and Lagrangians, and we get a (positive)
state if we start with a positive (non-ghost) cut propagator. We cannot expect
to get Lorentz invariant theories in general as we are working over a curved
spacetime, but if we work over Minkowski space and choose Lorentz invariant
cut propagators \ then we get Lorentz invariant free quantum field theories.
In the case of interacting theories Lorentz invariance is more subtle, even if
the Lagrangian is Lorentz invariant. Lorentz invariance depends on the
cancellation of infra-red divergences as we have to approximate the Lorentz
invariant Lagrangian by non Lorentz invariant Lagrangians with compact
support, and we can only show that infra-red divergences cancel up to inner
automorphisms. This allows for the possibility that the vacuum is not Lorentz
invariant, in other words Lorentz invariance may be spontaneously broken by
infra-red divergences, at least if the theory has massless particles. (It
seems likely that if there are no massless particles then infra-red
divergences cancel and we recover Lorentz invariance, but I have been too lazy
to check this in detail.)

In the final section we discuss anomalies. Fujikawa {\cite{Fujikawa}} observed
that anomalies arise from the lack of invariance of Feynman measures under a
symmetry group, and we translate his observation into mathematical language.

The definitions above generalize to the relative case where spacetime is
replaced by a morphism $X \rightarrow Y$, whose fibers can be thought of as
spacetimes parameterized by $Y$. For example, the sheaf of densities $\omega$
is replaced by the dualizing sheaf or complex $\omega_{X / Y}$. We will make
no serious use of this generalization, though the section on regularization
could be thought of as an example of this where $Y$ is the spectrum of a ring
of meromorphic functions.

\section{The ultraviolet group}

We describe the structure of the ultraviolet group, and show that it acts
simply transitively on the Feynman measures associated with a given
propagator.

\begin{theorem}
  \label{uvgroupstructure}The map taking a renormalization $\rho : S \omega
  SJ \Phi \rightarrow S \omega SJ \Phi$ to its composition with
  the natural map \ $S \omega SJ \Phi \rightarrow S^1 \omega S^0 J \Phi
  = \omega$ identifies renormalizations with the elements of $\tmop{Hom} (S^{}
  \omega SJ \Phi, \omega)$ that vanish on $S^0 \omega SJ \Phi$
  and that are isomorphisms when restricted to $\omega = S^1 \omega S^0 J
  \Phi$.
\end{theorem}

\begin{proof}
  This is a variation of the dual of the fact that endomorphisms $\rho$ of a
  polynomial ring $R [x]$ correspond to polynomials $\rho (x)$, given by the
  image of the polynomial $x$ under the endomorphism $\rho$. It is easier to
  understand the dual result first, so suppose that $C$ is a cocommutative
  Hopf algebra and $\omega$ is a vector space (with $C$ acting trivially on
  $\omega$). Then the symmetric algebra $S \omega C = S (\omega \otimes C)$ is
  a commutative algebra acted on by $C$, and its endomorphisms (as a
  commutative algebra) correspond exactly to elements of $\tmop{Hom} (\omega,
  S \omega C)$ because any such map lifts uniquely to a $C$-invariant map from
  $\omega$ to $\omega C$, which in turn lifts to a unique algebra homomorphism
  from $S \omega C$ to itself by the universal property of symmetric algebras.
  This endomorphism is invertible if and only if the map from $\omega$ to
  $\omega = S^1 \omega C^0$ is invertible, where $C^0$ is the vector space
  generated by the identity of $C$.
  
  To prove the theorem, we just take the dual of this result, with $C$ now
  given by $SJ \Phi$. There is one small modification we need to make
  in taking the dual result: we need to add the condition that the element of
  $\tmop{Hom} (S \omega C, \omega)$ vanishes on $S^0 \omega C$ in order to get
  an endomorphism of $S \omega C$; this is related to the fact that
  endomorphisms of the polynomial ring $R [x]$ correspond to polynomials, but
  continuous endomorphisms of the power series ring $R [[x]]$ correspond to
  power series with vanishing constant term.
\end{proof}

The ultraviolet group preserves the increasing filtration $S^{\leqslant m}
\omega SJ \Phi$ and so has a natural decreasing filtration by the
groups $G_{\geqslant n}$, consisting of the renormalizations that fix all
elements of $S^{\leqslant n} \omega SJ \Phi$. The group $G =
G_{\geqslant 0}$ is the inverse limit of the groups $G / G_{\geqslant n}$, and
the commutator of $G_{\geqslant m}$ and $G_{\geqslant n}$ is contained in
$G_{\geqslant m + n}$, so in particular $G_{\geqslant 1}$ is an inverse limit
of nilpotent groups $G_{\geqslant 1} / G_{\geqslant n}$. The group \
$G_{\geqslant n}$ is a semidirect product $G_{\geqslant n + 1} G_n$ of its
normal subgroup $G_{\geqslant n + 1}$ with the group $G_n$, consisting of
elements represented by elements of $\tmop{Hom} (S^{} \omega SJ \Phi,
\omega)$ that are the identity on $S^1 \omega SJ \Phi$ if $n > 0$, and
vanish on $S^m \omega SJ \Phi$ for $m > 1$, $m \neq n + 1$.

\begin{lemma}
  The group $G$ is $\ldots G_2 G_1 G_0$ in the sense that any element of $G$
  can be written uniquely as an infinite product $\ldots g_2 g_1 g_0$ with
  $g_i \in G_i$, and conversely any such infinite product converges to an
  element of $G$.
\end{lemma}

\begin{proof}
  The convergence of this product follows from the facts that all elements
  $g_i$ preserve any space $S^{\leqslant m} \omega SJ \Phi$, and all
  but a finite number act trivially on it. The fact that any element can be
  written uniquely as such an infinite product follows from the fact that $G /
  G_{\geqslant n}$ is essentially the product $G_{n - 1} \ldots G_2 G_1 G_0$.
\end{proof}

The natural map
\[ S \Gamma \omega SJ \Phi \rightarrow \Gamma S \omega SJ \Phi
\]
is not an isomorphism, because on the left the symmetric algebra is taken over
the reals, while on the right it is essentially taken over smooth functions on
$M$.

\begin{lemma}
  The action of renormalizations on $\Gamma S \omega SJ \Phi$ lifts to
  an action on $S \Gamma_c \omega SJ \Phi$ that preserves the
  coproduct, the coaction of $\Gamma SJ \Phi$, and the product of
  elements with disjoint support.
\end{lemma}

\begin{proof}
  A renormalization is given by a linear map from $\Gamma_c S \omega SJ
  \Phi$ to $\Gamma_c \omega$, which by composition with the map $S \Gamma_c
  \omega SJ \Phi \rightarrow \Gamma S \omega SJ \Phi$ and the
  ``integration over $M$'' map $\Gamma_c \omega \rightarrow \tmmathbf{R}$
  lifts to a linear map from $S \Gamma_c \omega SJ \Phi$ to
  $\tmmathbf{R}$. This linear map has the special property that the product of
  any two elements with disjoint support vanishes, because it is multilinear
  over the ring of smooth functions. As in theorem \ref{uvgroupstructure}, the
  linear map gives an automorphism of $S \Gamma_c \omega SJ \Phi$
  preserving the coproduct and the coaction of $\Gamma SJ \Phi$. As the
  linear map vanishes on products of disjoint support, the corresponding
  renormalization preserves products of elements with disjoint support.
\end{proof}

In general, renormalizations do not preserve products of elements of $S
\Gamma_c \omega SJ \Phi$ that do not have disjoint support; the ones
that do are those in the subgroup $G_0$.

\begin{theorem}
  \label{transitive}The group of complex renormalizations acts simply
  transitively on the Feynman measures associated with a given cut local
  propagator.
\end{theorem}

\begin{proof}
  We first show that renormalizations $\rho$ act on Feynman measures $\omega$
  associated with a given local cut propagator. We have to show that
  renormalizations preserve nondegeneracy, smoothness on the diagonal, and the
  Gaussian property. The first two of these are easy to check, because the
  value of $\rho (\omega)$ on any element is given by a finite sum of values
  of $\omega$ on other elements, so is smooth along the diagonal.
  
  To check that renormalizations preserve the Gaussian property
  \[ \text{$\omega (AB) = \sum \omega (A') \Delta (A_{}'', B'') \omega (B')$}
  \]
  we recall that renormalizations $\rho$ preserve products with disjoint
  support and also commute with the coaction of $SJ \Phi$. Since $A$
  and $B$ have disjoint supports we have $\rho (AB) = \rho (A) \rho (B)$.
  Since $\rho$ commutes with the coaction of $SJ \Phi$, the image of
  $\rho (A)$ under the coaction of $SJ \Phi$ is $\sum \rho (A') \otimes
  A''$, and similarly for $B$. Combining these facts with the Gaussian
  property for $\rho (A) \rho (B)$ shows that
  \[ \text{$\omega (\rho (AB)) = \sum \omega (\rho (A')) \Delta (A_{}'', B'')
     \omega (\rho (B'))$} \]
  or in other words the renormalization $\rho$ preserves the Gaussian
  property.
  
  To finish the proof, we have to show that for any two normalized smooth
  Feynman measures $\omega$ and $\omega'$ with the same cut local propagator,
  there is a unique complex renormalization $g$ taking $\omega$ to $\omega'$.
  We will construct $g = \ldots g_2 g_1 g_0$ as an infinite product, with the
  property that $g_{n - 1} \ldots g_0 \omega$ coincides with $\omega'$ on
  $e^{iL_F} S^{\leqslant n} \Gamma_c \omega SJ \Phi$. Suppose that
  $g_0, \ldots, g_{n - 1}$ have already been constructed. By changing $\omega$
  to \ $g_{n - 1} \ldots g_0 \omega$ we may as well assume that they are all
  1, and that $\omega$ and $\omega'$ coincide on \ $e^{iL_F} S^{\leqslant n}
  \Gamma_c \omega SJ \Phi$. We have to show that there is a unique $g_n
  \in G_n$ such that $g_n \omega$ and $\omega'$ coincide on $e^{L_F} S^{n + 1}
  \Gamma_c \omega SJ \Phi$.
  
  The difference \ $\omega - \omega'$, restricted to $e^{iL_F} S^{n + 1}
  \Gamma_c \omega SJ \Phi$, is a continuous linear function on
  $e^{iL_F} S^{n + 1} \Gamma_c \omega SJ \Phi \text{}$, which we think
  of as a distribution. Moreover, since both $\omega$ and $\omega'$ are
  determined off the diagonal by their values on elements of smaller degree by
  the Gaussian property, this distribution has support on the diagonal of
  $M^{n + 1}$. Since $\omega$ and $\omega'$ both have the property that their
  wave front sets on the diagonal are orthogonal to the diagonal, the same is
  true of their difference $\omega - \omega'$, so the distribution is given by
  a map $e^{iL_F} S^{n + 1} \Gamma_c \omega SJ \Phi \rightarrow
  \omega$. \ \ By theorem \ref{uvgroupstructure} this corresponds to some \
  renormalization $g_n \in G_n$, which is the unique element of $G_n$ such
  that $g_n \omega \tmop{and} \omega' \tmop{coincide} \tmop{on} e^{iL_F} S^{n
  + 1} \Gamma_c \omega SJ \Phi$ . 
\end{proof}

\section{Existence of Feynman measures}

We now prove theorem \ref{feynman existence} showing the existence of at least
one Feynman measure associated to any cut local propagator, by using
regularization and renormalization. Regularization means that we construct a
Feynman measure over a field of meromorphic functions, which will usually have
poles at the point we are interested in, and renormalization means that we
eliminate these poles by acting with a suitable meromorphic renormalization.

\begin{lemma}
  \label{bernsteinpolynomial} If $f_1, \ldots, f_m$ are polynomials in several
  variables, then there are non-zero (Bernstein-Sato) polynomials $b_i$ and
  differential operators $D_i$ such that
  \[ \text{ $b_i (s_1, \ldots, s_m)$} f_1 (z)^{s_1} \ldots f_m (z)^{s_m} = D_i
     (z) \left( f_i (z) f_1 (z)^{s_1} \ldots f_m (z)^{s_m}) \right. \]
\end{lemma}

\begin{proof}
  Bernstein's proof {\cite{Bernstein}} of this theorem for the case $m = 1$
  also works for any $m$ after making the obvious minor changes, such as
  replacing the field of rational functions in one variable $s_1$ by the field
  of rational functions in $m$ variables.
\end{proof}

\begin{corollary}
  \label{bernsteincontinuation}If $f_1, \ldots, f_m$ are polynomials in
  several variables then for any choice of continuous branches of the
  multivalued functions, $f_1 (z)^{s_1} \ldots f_m (z)^{s_m}$ can be
  analytically continued from the region where all $s_j$ have positive real
  part to a meromorphic distribution-valued function for all complex values of
  \ $s_1, \ldots, s_m$.
\end{corollary}

\begin{proof}
  This follows by using the functional equation of lemma
  \ref{bernsteinpolynomial} to repeatedly decrease each $s_j$ by 1, just as in
  Bernstein's proof {\cite{Bernstein}} for the case $m = 1$. 
\end{proof}

\begin{theorem}
  \label{regularization}Any cut local propagator $\Delta$ has a
  regularization, in other words a Feynman measure with values in a ring of
  meromorphic functions whose cut propagator at some point is $\Delta$.
\end{theorem}

\begin{proof}
  The following argument is inspired by the one in Etingof {\cite{Etingof}}.
  By using a locally finite smooth partition of unity, which exists since we
  assume that spacetime is metrizable, we can reduce to showing that a
  regularization exists locally. If a local propagator is smooth, it is easy
  to construct a Feynman measure for it, just by defining it as a sum of
  products of Feynman propagators. Now suppose that we have a meromorphic
  family of local propagators $\Delta^{_{}}_d$ depending on real numbers
  $d_i$, given in local coordinates by a finite sum of boundary values of
  terms of the form
  \[ s (x, y) p_1 (x, y)^{d_1} \ldots p_k (x, y)^{d_k} \log (p_{k + 1} (x, y))
     \ldots \]
  where $s$ is smooth in $x$ and $y$, and the $p_i$ are polynomials, and where
  we choose some branch of the powers and logarithms in each region where they
  are non-zero. In this case the Feynman measure can also be defined as a
  meromorphic function of $d$ for all real $d$. To prove this, we can forget
  about the smooth function $s$ as it is harmless, and we can eliminate the
  logarithmic terms by writing $\log (p)$ as $\frac{d}{dt} p^t \tmop{at} t =
  0$. For any fixed number of fields with derivatives of fixed order, the
  corresponding distribution is defined when all variables $d_i$ have
  sufficiently large real part, because the product of the propagators is
  smooth enough to be defined in this case. But this distribution is given \
  in local coordinates by the product the $d_i$'th powers of \ polynomials of
  $x$ and $y$. By Bernstein's corollary \ref{bernsteincontinuation} these
  products can be continued as a meromorphic distribution-valued function of
  the $d_i$ to all complex $d_i$.
  
  This gives a Feynman measure with values in the field of meromorphic
  functions in several variables, and by restricting functions to the diagonal
  we get a Feynman measure whose value are meromorphic functions in one
  variable.
\end{proof}

\begin{example}
  Dimensional regularization. Over Minkowski space of dimension $d$, there is
  a variation of the construction of a meromorphic Feynman measure, which is
  very similar to dimensional regularization. In dimensional regularization,
  one formally varies the dimension of spacetime, to get Feynman diagrams that
  are meromorphic functions of the dimension of spacetime. One way to make
  sense out of this is to keep the dimension of spacetime fixed, but vary the
  propagator of the free field theory, by considering it to be a meromorphic
  function of a complex number $d$. The propagator for a Hermitian scalar
  field, considered as a distribution of $z$ in Minkowski space, can be
  written as a linear combination of functions of the form
  \[ K_{d / 2 - 1} (c \sqrt{(z, z)}) / \sqrt{(z, z)^{}}^{d / 2 - 1} \]
  where $K_{\nu} (z)$ is a multi-valued modified Bessel function of the third
  kind, and where we take a suitable choice of branch (depending on whether we
  are considering a cut or a Feynman propagator). A similar argument using
  Bernstein's theorem shows that this gives a Feynman measure that is analytic
  in $d$ for $d$ with large real part and that can be analytically continued
  as a meromorphic function to all complex $d$. This gives an explicit example
  of a meromorphic Feynman measures for the usual propagators in Minkowski
  space.
\end{example}

\begin{theorem}
  \label{renormalizemeasure}Any meromorphic Feynman measure can be made
  holomorphic by acting on it with a meromorphic renormalization.
\end{theorem}

\begin{proof}
  This is essentially the result that a bare quantum field theory can be made
  finite by an infinite renormalization. \ Suppose that $\omega$ is a
  meromorphic Feynman measure. Using the same idea as in theorem
  \ref{transitive} we will construct a meromorphic renormalization $g =
  \text{$\ldots g_2 g_1 g_0$}$ as an infinite product, but this time we choose
  $g_n \in G_n$ \ to kill the singularities of order $n + 1$. \ The key point
  is to prove that these lowest order singularities are ``local'', meaning
  that they have support on the diagonal. (In the special case of
  translation-invariant theories on Minkowski spacetime \ this becomes the
  usual condition that they are ``polynomials in momentum'', or more precisely
  that their Fourier transforms are essentially polynomials in momentum on the
  subspace with total momentum zero). \ The locality follows from the Gaussian
  property of $\omega$, which determines $\omega$ at each order in terms of
  smaller orders except on the diagonal. In particular if $\omega$ is
  nonsingular at all orders at most $n$, then the singular parts of the order
  $n + 1$ terms all have support on the diagonal. Since the difference is
  smooth along the diagonal, we can find some $g_n \in G_n$ that kills off the
  order $n + 1$ singularities, \ as in \ theorem \ref{transitive}. Since
  renormalizations preserve the Gaussian property we can keep on repeating
  this indefinitely, killing off the singularities in order of their order.
\end{proof}

The famous problem of ``overlapping divergences'' is that the counter-terms
for individual Feynman diagrams used for renormalization sometimes contain
non-polynomial (logarithmic) terms in the momentum, which bring
renormalization to a halt unless they miraculously cancel when summed over all
Feynman diagrams. This problem is avoided in the proof above because by using
the ultraviolet group we only need to handle the divergences of lowest order
at each step, where it is easy to see that the \ logarithmic terms cancel.

\begin{theorem}
  \label{feynman existence}For any cut local propagator there is a Feynman
  measure associated to it.
\end{theorem}

\begin{proof}
  This follows from theorem \ref{regularization}, which uses regularization to
  show that there is a meromorphic Feynman measure, and theorem
  \ref{renormalizemeasure} which uses renormalization to show that the poles
  of this can be eliminated.
\end{proof}

\section{Subgroups of the ultraviolet group}

There are many additional desirable properties that one can impose on Feynman
measures, such as being Hermitian, or Lorentz invariant, or normal ordered,
and there is often a subgroup of the ultraviolet group that acts transitively
on the measures with the given property. We give several examples of this.

\begin{example}
  A Feynman measure can be normalized so that on $S^1 \Gamma_c \omega S^{^0}
  J \Phi = \Gamma_c \omega$ its value is given by integrating over spacetime
  (in other words $g = 1$ in definition \ref{feynman measure}), by acting on
  it by a unique element of the ultraviolet group consisting of
  renormalizations in $G_0$ that are trivial on $\omega S^{> 0} J \Phi$. This
  group can be identified with the group of nowhere-vanishing smooth complex
  functions on spacetime. The complementary normal subgroup of the ultraviolet
  group consists of the renormalizations that fix all elements of \ $\omega
  S^0 J \Phi = \omega$, and this acts simply transitively on the normalized
  Feynman measures. In practice almost any natural Feynman measure one
  constructs is normalized.
\end{example}

\begin{example}
  Normal ordering. In terms of Feynman diagrams, ``normal ordering'' means
  roughly that Feynman diagrams with an edge from a vertex to itself are
  discarded. We say that a Feynman measure is normally ordered if it vanishes
  on $\Gamma_c \omega S^{> 0} J \Phi$. Informally, \ $\omega S^{> 0} J \Phi$
  corresponds to Feynman diagrams with just one point and edges from this
  point to itself. We will say that a renormalization is normally ordered if
  it fixes all elements of $\omega S^{> 0} J \Phi$. The subgroup of normally
  ordered renormalizations acts transitively on the normally ordered Feynman
  measures. The group of all renormalizations is the semidirect product of its
  normal subgroup $G_{> 0}$ of normally-ordered renormalizations with the
  subgroup $G_0$ preserving all products. For any renormalization, there is a
  unique element of $G_0$ that takes it to a normally ordered renormalization.
  The Feynman measures constructed by regularization (in particular those
  constructed by dimensional regularization) are usually normally ordered if
  the spacetime has positive dimension, but are usually not for 0-dimensional
  spacetimes. This is because the propagators tend to contain a factor such as
  $(x - y)^{- 2 d}$ which vanishes for large $- d$ when $x = y$, and so
  vanishes on Feynman diagrams with just one point for all $d$ by analytic
  continuation. So for most purposes we can restrict to normally-ordered
  Feynman measures and normally-ordered renormalizations, at least for
  spacetimes of positive dimension. 
\end{example}

\begin{example}
  Normalization of Feynman propagators. In general a renormalization fixes the
  cut propagator but \ can change the Feynman propagator, by adding a
  distribution with support on the diagonal. However there is often a
  canonical choice of Feynman propagator: the one with a singularity on the
  diagonal of smallest possible order, which will often also be a Green
  function for some differential operator. We can add the condition that the
  Feynman propagator of a Feynman measure should be this canonical choice; the
  subgroup of renormalizations fixing the Feynman propagator, consisting of
  renormalizations fixing $S^2 \omega J \Phi$, acts simply transitively on
  these Feynman measures.
\end{example}

\begin{example}
  \label{simple operator}Simple operators. More generally, there is a subgroup
  consisting of renormalizations $\rho$ such that $\rho (aB) = \rho (a) \rho
  (B)$ whenever $a$ is simple (involving only one field), but where $B$ is
  arbitrary. This stronger condition is useful because it says (roughly) that
  simple operators containing only one field do not get renormalized; see the
  discussion in section \ref{interactingqft}. \ We can find a set of Feynman
  measures acted on simply transitively by this group by adding the condition
  that
  \[ \omega (aB) = \sum \Delta_F (aB_1) \omega (B_2) \]
  whenever $a$ is simple and $\sum B_1 \otimes B_2$ is the coproduct of $B$.
  This relation holds whenever $a$ and $B$ have disjoint supports by
  definition of a Feynman measure, so the extra condition says that it also
  holds even when they have overlapping supports. The key point is that the
  product of distributions above is always defined because any non-zero
  element of the wave front set of $\Delta_F$ is of the form $(p, - p)$. This
  would not necessarily be true if $a$ were not simple because we would get
  products of more than 1 Feynman propagator whose singularities might
  interfere with each other. In terms of Feynman diagrams, this says that
  vertices with just one edge are harmless: more precisely, with this
  normalization, adding a vertex with just one edge to a Feynman diagram has
  the effect of multiplying its value by the Feynman propagator of the edge.
  As this condition extends the Gaussian property to more Feynman diagrams, it
  can also be thought of as a strengthening of the translation invariance
  property of the Feynman measure.
\end{example}

\begin{example}
  \label{dyson}Dyson condition. Classically, Lagrangians were called
  renormalizable if all their coupling constants have non-negative mass
  dimension. The filtration on Lagrangian densities by mass dimension induces
  a similar filtration on Feynman measures and renormalizations. The Feynman
  measures of mass dimension $\leqslant 0$ are acted on simply transitively by
  the renormalizations of mass dimension $\leqslant 0$. This is useful,
  because the renormalizations of mass dimension at most 0 act on the spaces
  of Lagrangian densities of mass dimension at most 0, and these often form
  finite dimensional spaces, at least if some other symmetry conditions such
  as Lorentz invariance are added. For example, in dimension 4 the density has
  dimension $- 4$, so the (Lorentz-invariant) terms of the Lagrangian density
  of mass dimension at most 0 are given by (Lorentz invariant) terms of the
  Lagrangian of mass dimension at most 4, such as $\varphi^4$, $\varphi^2$,
  $\partial \varphi \partial \varphi$, and so on: the usual Lorentz-invariant
  even terms whose coupling constants have mass dimension at least 0. For
  example, we get a three-dimensional space of theories of the form $\lambda
  \varphi^4 + m \varphi^2 + z \partial \varphi \partial \varphi$ in this way,
  giving the usual $\varphi^4$ theory in 4 dimensions. 
\end{example}

\begin{example}
  Boundary terms. The Feynman measures constructed in section 3 have the
  property that they vanish on ``boundary terms''. This means that we quotient
  the space of local Lagrangians $\Gamma_c \omega SJ \Phi$ by its image
  under the action of smooth vector fields such as $\partial / \partial x_i$,
  or in other words we replace a spaces of $n$-forms by the corresponding de
  Rham cohomology group. These measures are acted on simply transitively by
  renormalizations corresponding to maps that vanish on boundary terms. This
  is useful in gauge theory, because some symmetries such as the BRST symmetry
  are only symmetries up to boundary terms. 
\end{example}

\begin{example}
  Symmetry invariance. Given a group (or Lie algebra) $G$ such as a gauge
  group acting on the sheaf $\Phi$ of classical fields and preserving a given
  cut propagator, the subgroup of $G$-invariant renormalizations acts simply
  transitively on the $G$-invariant Feynman measures with given cut
  propagator. In general there need not exist any $G$-invariant Feynman
  measure associated with a given cut local propagator, though if there is
  then $G$-invariant Lagrangians lead to $G$-invariant quantum field theories.
  The obstructions to finding a $G$-invariant measure are cohomology classes
  called anomalies, and are discussed further in section \ref{anomalies}.
\end{example}

\begin{example}
  Lorentz invariance. An important case of invariance under symmetry is that
  of Poincare invariance for flat Minkowski space. In this case the spacetime
  $M$ is Minkowski space, the Lie algebra $G$ is that of the Poincare group of
  spacetime translations and Lorentz rotations, and the cut propagator is one
  of the standard ones for free field theories of fields of finite spin. Then
  dimensional regularization is invariant under $G$, so we get a Feynman
  measure invariant under the Poincare group, and in particular there are no
  anomalies for the Poincare algebra. The elements of the ultraviolet group
  that are Poincare invariant act simply transitively on the Feynman measures
  for this propagator that are Poincare invariant. If we pick any such
  measure, then we get a map from invariant Lagrangians to invariant quantum
  field theories. 
\end{example}

\begin{example}
  Hermitian conditions. The group of complex renormalizations has a real form,
  consisting of the subgroup of (real) renormalizations. This acts simply
  transitively on the Hermitian Feynman measures associated with a given cut
  local propagator. The Hermitian Feynman measures (or propagators) are not
  the real-valued ones, but satisfy a more complicated Hermitian condition
  described in definition \ref{hermitianmeasure}.
\end{example}

\section{\label{freeqft}The free quantum field theory}

We extend the Feynman measure $\omega : e^{iL_F} S \Gamma_c \omega SJ
\Phi \rightarrow \tmmathbf{C}$, which is something like a measure on classical
fields, to $\omega : Te^{iL_F} S \Gamma_c \omega SJ \Phi \rightarrow
\tmmathbf{C}$. This extension, restricted to the even degree subalgebra $T_0
e^{iL_F} S \Gamma_c \omega SJ \Phi$, is the free quantum field theory.
We check that it satisfies analogues of the Wightman axioms.

Formulas involving coproducts can be confusing to write down and manipulate.
They are much simpler for the ``group-like'' elements $g$ satisfying $\Delta
(g) = g \otimes g$, $\eta (g) = 1$, which form a group in any cocommutative
Hopf algebra. One problem is that most of the Hopf algebras we use do not have
enough group-like elements over fields: in fact for symmetric algebras the
only group-like element is the identity. However they have plenty of
group-like elements if we add some nilpotent elements to the base field, such
as $\exp (\lambda a)$ for any primitive $a$ and nilpotent $\lambda$ (in
characteristic 0). We will adopt the convention that when we talk about
group-like elements, we are tacitly allowing extensions of the base ring by
nilpotent elements.

Recall that $Te^{iL_F} S \Gamma_c \omega SJ \Phi$ is the tensor algebra of
$e^{iL_F} S \Gamma_c \omega SJ \Phi$, with the product denoted by $\otimes$ to
avoid confusing it with the product of $S \Gamma_c SJ \Phi$. We denote
the identity of $S \Gamma_c SJ \Phi$ by $1$, and the identity of
$\tmop{TS} \Gamma_c \omega SJ \Phi$ by $1_T$. The involution $\ast$ is
defined by $(A_1 \otimes \ldots \otimes A_n)^{\ast} = A_n^{\ast} \otimes
\ldots \otimes A_1^{\ast}$, and $\ast$ is $- 1$ on $\Gamma_c \omega SJ
\Phi$.

\begin{theorem}
  \label{extendomega}If $\omega : e^{iL_F} S \Gamma_c \omega SJ \Phi
  \rightarrow C$ is a Feynman measure then there is a unique extension of \
  $\omega$ to $Te^{iL_F} S \Gamma_c \omega SJ \Phi$ such that
  \begin{itemizedot}
    \item Gaussian condition: if $A, B_1, \ldots, B_m$ are group-like then
    \begin{eqnarray*}
      &  & e^{- iL_F} \omega (A_{} \otimes B_m \otimes \ldots \otimes B_1)\\
      & = & \sum e^{- iL_F} \omega (A_{} \otimes 1 \otimes \ldots \otimes 1)
      \Delta (A, B_m \ldots B_1) e^{- iL_F} \omega (B_m \otimes \ldots \otimes
      B_1)
    \end{eqnarray*}
    \ \ Both sides are considered as densities, as in definition \ref{feynman
    measure}.
    
    \item $e^{- iL_F} \omega (A \otimes A \otimes 1 \otimes \ldots \otimes 1)
    = 1$ \ for $A$ group-like (Cutkosky condition; see{\cite[section
    6]{hooft}}.)
  \end{itemizedot}
\end{theorem}

\begin{proof}
  We first check that all the products of distributions are well defined by
  examining their wave front sets. All the distributions appearing have the
  property that their wave front sets have no positive or negative elements.
  This follows by induction on the complexity of an element: if all smaller
  elements have this property, it implies that the products defining it are
  well defined, and also implies that it has the same property.
  
  Existence and uniqueness of $\omega$ follows because the Cutkosky condition
  defines it on elements of the form $A \otimes 1 \otimes 1 \otimes \ldots
  \otimes 1$ in terms of those of the form $A \otimes 1 \otimes \ldots \otimes
  1$, and the Gaussian condition then determines it on all elements. \ 
\end{proof}

We can also define $\omega$ directly as follows. When the propagator is
sufficiently regular then the Gaussian condition means that we can write
$\omega$ on $e^{iL_F} S \Gamma_c \omega SJ \Phi$ as a sum over all ways
of joining up the \ fields of an element of $e^{iL_F} S \Gamma_c \omega
SJ \Phi$ in pairs, where we take the propagator of each pair and
multiplying these together. This is of course essentially the usual sum over
Feynman diagrams. A minor difference is that we do not distinguish between
``internal'' vertices associated with a Lagrangian and integrated over all
spacetime, and ``extenal'' vertices associated with a field and integrated
over a compact set: all vertices are associated with a composite operator that
may be a Lagrangian or a simple field or a more general composite operator,
and all vertices are integrated over compact sets as all coefficients are
assumed to have compact support.

Similarly we can \ define the extension of $\omega$ to \ $Te^{iL_F} S
\Gamma_c \omega SJ \Phi$ by writing the distributions defining $\omega$
\ as a sum over more complicated Feynman diagrams whose vertices are in
addition labeled by non-negative integers, such that
\begin{itemize}
  \item The propagators from $A_i$ to $A_i$ are Feynman propagators.
\end{itemize}
\begin{itemize}
  \item The propagators from $A_i$ to $A_j$ for $i < j$ are cut propagators
  $\Delta$, with positive wave front sets on $i$ and negative wave front sets
  on $j$.
  
  \item The diagram is multiplied by a factor of $(- 1)^{\deg (A_2 A_4 A_6
  \ldots)}$; in other words we apply $\ast$ to $A_2$, $A_4, \ldots$.
\end{itemize}
In general if \ the propagator is not sufficiently regular (so that products
of propagators might not be defined when some points coincide) we can
construct $\omega$ by regularization and renormalization as in section 3,
which preserves the conditions defining $\omega$.

Now we show that $\omega$ satisfies the locality property of quantum field
theories (operators with spacelike-separated supports commute) by showing that
it vanishes on the following locality ideal.

\begin{definition}
  \label{localityideal}$T_0 S \Gamma_c \omega SJ \Phi$ is the
  subalgebra of even degree elements of $T_{} S \Gamma_c \omega SJ
  \Phi$. The locality ideal is the 2-sided ideal of $T_0 S \Gamma_c \omega
  SJ \Phi$ spanned by the coefficients of elements of the form
  \[ \ldots \otimes Y_1 \otimes ABD \otimes DBC \otimes X_n \otimes \ldots
     \otimes X_1 - \ldots \otimes Y_1 \otimes AD \otimes DC \otimes X_n
     \otimes \ldots \otimes X_1 \]
  (for $A, C \in S \Gamma_c \omega SJ \Phi$ \ and $B, D \in S \Gamma_c
  \omega SJ \Phi [[\tmmathbf{\lambda}]]$ with $B, D$ group-like) if $n$
  is even and there are no points in the support of $B$ that are
  $\leqslant$any points in the support of $A$ or $C$, or if $n$ is odd and \
  there are no points in the support of $B$ that are $\geqslant$any points in
  the support of $A$ or $C$.
  
  The algebra \ $T_0 e^{iL_F} S \Gamma_c \omega SJ \Phi$ and its
  locality ideal are defined in the same way. 
\end{definition}

\begin{remark}
  The map $\omega$ on $T_0 e^{iL_F} S \Gamma_c \omega SJ \Phi$ depends
  on the choice of Feynman measure. We can define a canonical map independent
  of the choice of Feynman measure by taking the underlying $*$-algebra to have
  elements represented by pairs $(\omega, A)$ for a Gaussian measure $\omega$
  and $A \in \text{$T_0 e^{iL_F} S \Gamma_c \omega SJ \Phi$}$, where we
  identify $(\omega, A)$ with $(\rho \omega, \rho A)$ for any renormalization
  $\rho$. The canonical state, also denoted by $\omega$, then takes an element
  represented by $(\omega, A)$ to $\omega (A)$.
\end{remark}

\begin{theorem}
  \label{localideal}$\omega$ vanishes on the locality ideal.
\end{theorem}

\begin{proof}
  We use the notation of \ definition \ref{localityideal}. We prove this for
  elements with $n$ even; the case $n$ odd is similar. We can assume that the
  propagator $\Delta$ is sufficiently regular, as we can obtain the general
  case from this by regularization and renormalization. We will first do the
  special case when $D = 1$. We can assume that $B = b_1 \ldots b_k$ is
  homogeneous of some order $k$ and write $B_I$ for $\prod_{j \in I} b_j$. If
  $k = 0$ then the result is obvious as $B$ is constant and both sides are the
  same, so we can assume that $k > 0$. We show that if $k > 0$ then $\omega$
  vanishes on
  \[ \sum_{I \cup J =\{1, \ldots k\}} (- 1)^{|I|} \ldots \otimes Y_1 \otimes
     AB_I \otimes B_J C \otimes X_n \otimes \ldots \otimes X_1 \]
  by showing that the \ terms cancel out in pairs. This is because if $j$ is
  the index for which the support of $b_j$ is maximal then $\omega$ has the
  same value on
  
  $\ldots \otimes Y_1 \otimes AB_I b_j \otimes B_J C \otimes X_n \otimes
  \ldots \otimes X_1$
  
  and
  
  $\ldots \otimes Y_1 \otimes AB_I \otimes b_j B_J C \otimes X_n \otimes
  \ldots \otimes X_1$
  
  Now we do the case of general $D$. We can assume that the support of $D$ is
  either $\leqslant$ all points of the support of $B$ or there are no points
  of it that are $\leqslant$ any points in the support of $A$ or $C$. In the
  first case the result follows from the special case $D = 1$ by replacing $A$
  and $C$ by $AD$ and $CD$. In the second case it follows from 2 applications
  of the special case $D = 1$, replacing $B$ by $D$ and $BD$, that both terms
  are equal to $\ldots \otimes Y_1 \otimes A \otimes C \otimes X_n \otimes
  \ldots \otimes X_1$ and are therefore equal. 
\end{proof}

This proof, in the special case that $\omega$ vanishes on $B \otimes B - 1
\otimes 1$ for $B$ group-like, is more or less the proof of unitarity of the
S-matrix using the ``largest time equation'' given in {\cite[section
6]{hooft}}. The locality ideal is not the largest ideal on which $\omega$
vanishes, as $\omega$ also vanishes on $A \otimes 1 \otimes 1 \otimes B - A
\otimes B$; in other words we can cancel pairs $1 \otimes 1$ wherever they
occur.

\begin{theorem}
  \label{commute mod local}Elements of $T_0 S \Gamma_c SJ \Phi$ with
  spacelike-separated supports commute modulo the locality ideal.
\end{theorem}

\begin{proof}
  It is sufficient to prove this for group-like degree 2 elements, as if two
  even degree elements have spacelike-separated supports then they are
  polynomials in degree 2 elements with spacelike separated supports. We will
  work modulo the locality ideal. Suppose that the supports of the group-like
  elements $W \otimes X \otimes Z$ and $Y$ are spacelike-separated. Then \
  applying \ theorem \ref{localideal} twice gives
  \[ W \otimes X \otimes YZ = W Y \otimes X Y \otimes Y Z = W Y \otimes X
     \otimes Z \]
  Applying this 4 times for various values of $W$, $X$, $Y$, and $Z$ shows
  that if $A \otimes B$ and $C \otimes D$ are group-like and have spacelike
  separated supports, then
  \[ A \otimes B \otimes C \otimes D = AC \otimes B \otimes I \otimes D = AC
     \otimes I \otimes I \otimes BD = AC \otimes D \otimes I \otimes B = C
     \otimes D \otimes A \otimes B \]
  $\tmop{so} A \otimes B$ and $C \otimes D$ commute.
\end{proof}

Now we study when the quantum field theory $\omega$ is Hermitian, and show
that we can find a Hermitian quantum field theory associated to any Hermitian
local cut propagator, and show that the group of real renormalizations acts
transitively on them.

\begin{definition}
  \label{hermitianmeasure}We say that a Feynman measure $\omega$ is Hermitian
  if its extension to $T_{} S \Gamma_c \omega SJ \Phi$ is Hermitian
  when restricted to the even subalgebra $T_0 S \Gamma_c \omega SJ
  \Phi$.
\end{definition}

\begin{lemma}
  If the local cut propagator $\Delta$ is Hermitian, then it has a Hermitian
  Feynman measure associated with it. \ 
\end{lemma}

\begin{proof}
  We can assume that the regularization of $\Delta$ is also Hermitian, by
  replacing it by the average of itself and its Hermitian conjugate. We can
  check directly that the meromorphic family of Feynman measures associated to
  this Hermitian regularization is Hermitian on $T_0 S \Gamma_c \omega
  SJ \Phi$ (but not on the whole of $T_{} S \Gamma_c \omega SJ
  \Phi$); in other words $\omega (A_n \otimes \ldots \otimes A_1) = \omega
  (A_1^{\ast} \otimes \ldots \otimes A_n^{\ast})^{\ast}$ if $n$ is even. For
  example, we get a sign factor of $- 1^{\deg (A_2) + \deg (A_4) + \ldots}$ in
  the definition of $\omega$ on the first term, a sign factor of \ $- 1^{\deg
  (A_1) + \deg (A_3) + \ldots}$ form the definition of $\omega$ for the second
  term, whose quotient is the factor $- 1^{\deg (A_1) + \deg (A_2) + \ldots}$
  coming from the action of $\ast$ on $A_n \otimes \ldots \otimes A_1$ because
  $n$ is even. We can then renormalize using real renormalizations to
  eliminate the poles, and the resulting Feynman measure will be Hermitian.
\end{proof}

\begin{lemma}
  If a Feynman measure $\omega$ is Hermitian and $\rho$ is a complex
  renormalization, then $\rho (\omega)$ is Hermitian if and only if $\rho$ is
  real. In particular the subgroup of (real) renormalizations acts simply
  transitively on the Hermitian Feynman measures associated with a given cut
  local propagator. 
\end{lemma}

\begin{proof}
  This follows from $\rho (\omega)^{\ast} = \rho^{\ast} (\omega^{\ast})$, and
  the fact that complex renormalizations act simply transitively on Feynman
  measures associated with a given cut local propagator. 
\end{proof}

Next we show that $\omega$ is a state \ (in other words the space of physical
states is positive definite) when the cut propagator $\Delta$ is positive, by
using a representation of the physical states as a space of distributions. We
define the space $H_n$ of $n$-particle states to be the space of continuous
linear maps $S^n \Gamma \omega \Phi \rightarrow \tmmathbf{C}$ (considered as
compactly supported symmetric distributions on $M^n$) whose wave front sets \
have no positive or negative elements, with a sesquilinear form given by \
\[ \langle a, b \rangle = \int_{x, y \in M^n} a (x_1, \ldots) \prod_j \Delta
   (x_j, y_j) b (y_j, \ldots)^{\ast} dxdy. \]
This is similar to the usual definition of the inner product on the space of
states of a free field theory, except that we are using distributions rather
than smooth functions. We check this is well defined. To show the product of
distributions in the integral is defined we need to check that no sum of
non-zero elements of the wave front sets is zero, and this follows because
nonzero elements of the wave front set of the product of propagators are of
the form $(p, q)$ with $p > 0$ and $q < 0$, but $a$ and $b$ by assumption have
no positive or negative elements in their wave front sets. \ The integral over
$M^n$ is well defined because $a$ and $b$ have compact support.

\begin{lemma}
  There is a map $f$ from $T_0 S \Gamma_c \omega SJ \Phi$ to the
  orthogonal direct sum$\oplus H_{_n}$ with
  \[ \omega (AB) = \langle f (A^{\ast}), f (B) \rangle . \]
\end{lemma}

\begin{proof}
  By theorem \ref{extendomega}, $\omega (AB)$ is given by
  \[ \sum \omega (A') \Delta (A'', B'') \omega (B') \]
  where $\sum A' \otimes A''$ is the image of $A$ under the coaction of
  $\Gamma_c SJ \Phi$. This is equal to $\langle f (A^{\ast}), f (B)
  \rangle$ if we define $f (A)$ as follows. Suppose that $A = A_{11} A_{12}
  \ldots \otimes A_{21} A_{22} \ldots$., and let the image of $A_{jk}$ under
  the coaction of $\Gamma_c SJ \Phi$ be $\sum A_{jk}' \otimes
  A_{jk}''$. Then $\omega ( A'_{11} A_{12}' \ldots \otimes A_{21}' A_{22}'
  \ldots$.) can be regarded as a distribution on $M^n$, where $n$ is the total
  number of elements $A_{jk}$. On the other hand, $A''_{11} A''_{12} \ldots
  A''_{21} A''_{22} \ldots$ is a function on $M^m$, where $m$ is the sum of
  the degree of the elements $A''_{jk}$, in other words the number of fields
  occurring in them. There is also a map from $m$ to $n$, which induces a map
  from $M^n$ to $M^m$, and so by push-forward of densities a map from
  densities on $M^n$ to densities on $M^m$. The image $f (A)$ is then given by
  taking the push-forward from $M^n$ to $M^m$ of the compactly supported
  distribution $\omega ( A'_{11} A_{12}' \ldots \otimes A_{21}' A_{22}'
  \ldots$.) on $M^n$, multiplying \ by the function $A''_{11} A''_{12} \ldots
  A''_{21} A''_{22} \ldots$ on $M^m$, symmetrizing the result, \ and repeating
  this for each summand of \ $\sum A_{jk}' \otimes A_{jk}''$.
\end{proof}

\begin{corollary}
  If the cut local propagator $\Delta$ is positive, then $\omega : Te^{iL_F} S
  \Gamma_c \omega SJ \Phi \rightarrow \tmmathbf{C}$ is a state.
\end{corollary}

\begin{proof}
  This follows from the previous lemma, because if $\Delta$ is positive then
  so is the sesquilinear form $\langle, \rangle$ on $H_n$, and therefore
  $\omega (A^{\ast} A) = \langle f (A), f (A) \rangle \geqslant 0$.
\end{proof}

\section{\label{interactingqft}Interacting quantum field theories}

We construct the quantum field theory of a Feynman measure and a compactly
supported Lagrangian, by taking the image of the free field theory $\omega$
under an automorphism $e^{iL_I}$ where $L_I$ is the interaction part of the
Lagrangian. This automorphism is only well defined if the interaction
Lagrangian $L_I$ has infinitesimal coefficients, so the interacting quantum
field theories we construct are perturbative theories taking values in rings
of formal power series $\tmmathbf{C}[\tmmathbf{\lambda}]
=\tmmathbf{C}[\lambda_1, \ldots]$ in the coupling constants $\lambda_1,
\ldots$. (By ``infinitesimal'' we mean elements of formal power series rings
with vanishing constant term.) We then lift the construction to all actions
(possible without compact support) by showing that infra-red divergences
cancel up to inner automorphisms.

\begin{lemma}
  The Hopf algebra $S \Gamma_c \omega SJ \Phi$ acts on \ the algebra
  $\text{$T_0 S \Gamma_c \omega SJ \Phi$}$, and maps the locality ideal
  to itself. Group-like Hermitian elements of the Hopf algebra $S \Gamma_c
  \omega SJ \Phi [[\tmmathbf{\lambda}]]$ preserve the subset of
  positive elements, and therefore act on the space of states of $T_0 S
  \Gamma_c \omega SJ \Phi [[\tmmathbf{\lambda}]]$.
\end{lemma}

\begin{proof}
  Group-like elements are algebra automorphisms, and if they are also
  Hermitian they commute with the involution $\ast$. In particular group-like
  Hermitian elements preserve the set of positive elements (generated by
  positive linear combinations of elements of the form $a^{\ast} a$), and so
  map positive linear forms to positive linear forms.
\end{proof}

\begin{definition}
  The quantum field theory of a Lagrangian $L = L_F + L_I$, where $L_I$ has
  compact support and infinitesimal coefficients, is $e^{- iL} \omega : T_0 S
  \Gamma_c \omega SJ \Phi \rightarrow
  \tmmathbf{C}[[\tmmathbf{\lambda}]]$.
\end{definition}

The Hopf algebra $S \Gamma_c \omega SJ \Phi$ acts on the vector space
$S \Gamma_c \omega SJ \Phi$ by multiplication, so group-like elements
of the form $e^{iL_F + iL_I}$ take $S \Gamma_c \omega SJ \Phi$ to
$e^{iL_F} S \Gamma_c \omega SJ \Phi$ and $T_0 S \Gamma_c \omega
SJ \Phi$ to $T_0 e^{iL_F} S \Gamma_c \omega SJ \Phi$. Since
$\omega$ is in the dual of $T_0 e^{iL_F} S \Gamma_c \omega SJ \Phi$,
this shows that $e^{- iL} \omega$ is in the dual of $T_0 S \Gamma_c \omega SJ
\Phi$.

\begin{corollary}
  (Locality) Elements of $T_0 S \Gamma_c \omega SJ \Phi$ with
  spacelike-separated supports commute when acting on the space of physical
  states of $e^{- iL} \omega$.
\end{corollary}

\begin{proof}
  By theorem \ref{localideal} the operators of the locality ideal act
  trivially on the space of physical states of $\omega$. Since $e^{- iL}$
  preserves the locality ideal, the locality ideal also acts trivially on the
  space of physical states of $e^{- iL} \omega$. \ By lemma \ref{commute mod
  local} this implies that operators with spacelike separated supports commute
  on this space. 
\end{proof}

This constructs the quantum field theory of a Lagrangian whose interaction
part has compact support (and is infinitesimal). We now extend this to the
case when the interaction part need not have compact support. We do this by
using a cutoff function to give the Lagrangian compact support, and then we
then try to show that the result is independent of the choice of cutoff
function, provided it is 1 in a sufficiently large region. To do this we need
to assume that spacetime is globally hyperbolic, and we also find that the
result is not quite independent of the choice of cutoff.

If $f$ is a smooth function on $M$ then multiplication by $f$ is a linear
transformation of $\Gamma \omega SJ \Phi$ and therefore induces a
homomorphism of $S \Gamma \omega SJ \Phi$, denoted by $A \rightarrow
A^f$. If $A = e^{iL}$ is group-like, then $A^f = e^{iLf}$. If $f$ has compact
support then so does $A^f$ so that $A^f \omega$ is defined. We try to extend
the definition of $A^f \omega$ to more general functions $f$ in the hope that
we can take $f$ to be close to 1.

\begin{lemma}
  Suppose that $f$ and $g$ are compactly supported smooth functions on $M$ and
  $n$ is even. If $f = g$ on the past of $A_1 \ldots A_n$ then (modulo the
  locality ideal)
  \[ \begin{array}{lll}
       e^{- iL_F} A^f \omega (A_n \otimes \ldots \otimes A_1) & = & e^{- iL_F}
       A^g \omega (A_n \otimes \ldots \otimes A_1)
     \end{array} \]
  If \ $f = g$ on the future of $A_1 \ldots A_n$ then
  \[ \begin{array}{lll}
       e^{- iL_F} A^f \omega (A_n \otimes \ldots \otimes A_1) & = & e^{- iL_F}
       A^g \omega (A^{g - f} \otimes 1 \otimes A_n \otimes \ldots \otimes A_1
       \otimes 1 \otimes A^{g - f})
     \end{array} \]
\end{lemma}

\begin{proof}
  We work modulo the locality ideal. The first equality follows from
  \[ A^{- f} A_n \otimes \ldots \otimes A^{- f} A_1 = A^{- g} A_n \otimes
     \ldots \otimes A^{- g} A_1 \]
  which in turn follows from theorem \ref{localideal} by repeatedly inserting
  $A^{f - g} \otimes A^{f - g}$ (using the fact that $n$ is even). The second
  equality follows in the same way from
  \begin{eqnarray*}
    &  & A^{- f} \otimes A^{- f} \otimes A^{- f} A_n \otimes \ldots \otimes
    A^{- f} A_1 \otimes A^{- f} \otimes A^{- f}\\
    & = & A^{- f} \otimes A^{- g} \otimes A^{- g} A_n \otimes \ldots \otimes
    A^{- g} A_1 \otimes A^{- g} \otimes A^{- f}
  \end{eqnarray*}
\end{proof}

This lemma shows that the restriction of $A^f \omega$ to arguments with
support in some fixed compact subset of $M$ is almost independent of the
choice of $f$ provided that $f$ is 1 on the convex hull of the argument:
different choices of $f$ are related by a locally inner automorphism of $T_0 S
\Gamma_c \omega SJ \Phi$, given by conjugation by elements of the form
$1 \otimes A^h$. If the spacetime is globally hyperbolic in the sense that the
convex hull of a compact set is contained in a compact set, then we can always
find a suitable $f$ that is 1 on the convex hull $X$ of the argument, so we
can construct the interacting quantum field theory. The result does not depend
on the choice of cutoff $f$ on the future of $X$, but does depend slightly on
the choice of cutoff in the past of $X$. The choice of cutoff in the past
corresponds to choices of the vacuum: roughly speaking, we turn off the
interaction in the distant past, which gives different vacuums. More
precisely, if we have two different cutoffs $f$ and $g$ then their vacuums,
which are the images of $e^{i (L_F + fL_I)}$ and $e^{i (L_F + gL_I)}$ will
differ by a factor of $e^{i (f - g) L_I}$. This does not change the observable
physics, beause all these choices of cutoffs give isomorphic quantum field
theories. However it does cause difficulties in constructing a Lorentz
invariant theory, because the choice of cutoff in the past is not Lorentz
invariant, so the vacuums are also not Lorentz invariant, or in other words
Lorentz invariance may be spontaneously broken. \ Presumably in theories with
a mass gap one can take the limit as the cutoff in the past tends to time $-
\infty$ and get a Lorentz invariant vacuum, but in theories with massless
particles such as QED there is an obstruction to constructing a Lorentz
invariant vacuum: Lorentz invariance might be spontaneously broken by infrared
divergences. This is a well known problem, which is not worth worrying about
too much, because the physical universe is not globally Lorentz invariant.

The time-ordered operator $T (A)$ of an element $A \in S \Gamma_c \omega
SJ \Phi$ is \ defined to be $1 \otimes A$. This has the property that
\[ T (A_n \ldots A_1) = 1 \otimes A_n \ldots A_1 = 1 \otimes A_n \otimes
   \ldots \otimes 1 \otimes A_1 = T (A_n) \ldots T (A_1) \]
whenever the composite fields $A_i \in \Gamma_c \omega SJ \Phi$ are in
order of increasing time of their supports. This formula is sometimes used as
a ``definition'' of the time-ordered product $T (A_n \ldots A_1)$, though this
does not define it when some of the factors have overlapping supports, and in
general the time-ordered product depends on the choice of Feynman measure
$\omega$. The scattering matrix $S$ of the quantum field theory is $S = T
(e^{iL_I}) = 1 \otimes e^{iL_I}$; this is essentially the LSZ \ reduction
formula of Lehmann, Symanzik, and Zimmermann {\cite{Lehmann}}.

We now show that if we change the Feynman measure, then we still get an
isomorphic quantum field theory provided we make a suitable change in the
Lagrangian. If we change $\omega$ to a different Feynman measure for the same
cut local propagator, these will differ by a unique renormalization $\rho$; in
other words the other Feynman measure will be $\rho \omega$. The quantum field
theory $e^{- iL} \omega$ changes under this renormalization of $\omega$ by
\begin{eqnarray*}
  e^{- iL} \omega (A_1 \otimes \ldots) & = & \omega (e^{iL} A_1 \otimes
  \ldots)\\
  & = & \rho (\omega) (\rho (e^{iL} A_1) \otimes \ldots)\\
  & = & \rho (e^{- iL}) \rho (\omega) (\rho (e^{- iL}) \rho (e^{iL} A_1)
  \otimes \ldots)
\end{eqnarray*}
so the quantum field theory stays the same under renormalization by $\rho$ if
we transform the Lagrangian by
\[ iL \rightarrow \log (\rho (\exp (iL)), \]
which is a nonlinear transformation because renormalizations need not commute
with products or exponentiation, and change the operators $A_n$ by
\[ A_n \rightarrow \rho (e^{- iL}) \rho (e^{iL} A_n) . \]
If $A_n$ is a simple operator and $\rho$ satisfies the condition of example
\ref{simple operator} then $\rho (e^{iL} A_n) = \rho (e^{iL}) \rho (A_n) =
\rho (e^{iL}) A_n$, so in this special case $A_n$ is unchanged, or in other
words simple operators are not renormalized. The behavior of composite
operators under renormalization can be quite complicated when expanded out in
terms of fields. The usual Wightman distributions used to construct a quantum
field theory use only simple operators, so the only effect of renormalization
on Wightman distributions comes from the nonlinear transformation of the
Lagrangian. This nonlinear transformation of Lagrangians is the usual action
of renormalizations on Lagrangians used in physics texts to convert an
infinite ``bare'' Lagrangian $L$ to a finite physical one $L_0$; the bare and
physical Lagrangians are related by $iL_0 = \log (\rho (\exp (iL))$, where
$\rho$ is an infinite renormalization taking an infinite Feynman measure, such
as the one given by dimensional regularization, to a finite one.

The orbit of a Lagrangian under this nonlinear action of the ultraviolet
group is in general infinite dimensional. It can sometimes be cut down to a
finite dimensional space as follows. As in example \ref{dyson}, we cut down to
the group of renormalizations of mass dimension at most 0, which acts on the
space of Lagrangians whose coupling constants all have mass dimension at least
0. If we also add the condition that the Lagrangian is Lorentz invariant, then
we sometimes get finite dimensional spaces of Lagrangians. The point is that
the classical fields themselves tend to have positive mass dimension, so if
the coupling constants all have non-negative mass dimension then the fields
appearing in any term of the Lagrangian have total mass at most $d$
(cancelling out the $- d$ coming from the density) which severely limits the
possibilities. At one time the Lagrangians with all coupling constants of
non-negative mass dimension were called renormalizable Lagrangians, though now
all Lagrangians are regarded as renormalizable in a more general sense where
one allows an infinite number of terms in the Lagrangian.

\section{\label{anomalies}Gauge invariance and anomalies}

If a Lagrangian is invariant under some group, this does not imply that the
quantum field theories we construct from it are also invariant, because as
Fujikawa {\cite{Fujikawa}} pointed out we also need to choose a Feynman
measure and there may not be an invariant way of doing this. The obstructions
to finding an invariant quantum field theory lie inside certain cohomology
groups and are called anomalies. We show that if these anomalies vanish then
we can construct invariant quantum field theories.

Suppose that a group $G$ acts on $SJ \Phi$ and preserves the set of
Feynman measures with given cut local propagator, and suppose that we have
chosen one such Feynman measure $\omega$. In practice we often start with an
action of a Lie algebra or superalgebra, such as that generated by the BRST
operator, which can be turned into a group action in the usual way by working
over a ring with nilpotent elements. If $g \in G$ then $g \omega$ is another
Feynman measure with the same propagator, so
\[ \omega = \rho_g g \omega \]
for a unique renormalization $\rho_g$. This defines a non-abelian 1-cocycle:
$\rho_{gh} = \rho_g g (\rho_h)$, where $g (\rho_h) = g \rho_h g^{- 1}$. Since
$\omega$ is invariant under $\rho_g g$, we find that
\[ \omega (e^{iL} A_1) = \omega (\rho_g g (e^{iL} A_1)) = \omega (e^{iL} e^{-
   iL} \rho_g g (e^{iL} A_1)) \]
so that $e^{- L} \omega$ is invariant under the transformation taking
arguments $A_1$ to $e^{- iL} \rho_g g (e^{iL} A_1)$. This transformation fixes
$1$ if $e^{iL}$ is fixed by $\rho_g g$. If in addition $\rho_g g (e^{iL} A_1)
= \rho_g g (e^{iL}) \rho_g g (A_1)$ (which is not automatic as $\rho_g$ need
not preserve products) then $A_1$ is taken to $\rho_g g (A_1)$ by this
transformation.

This shows that we really want a Lagrangian $L$ such that $e^{iL}$ is
invariant under the modified action $e^{iL} \rightarrow \rho_g g (e^{iL})$.
This is not the same as asking for $\rho_g g (iL) = iL$ because $\rho_g$ need
not preserve products (although $g$ usually does). In practice we usually have
a Lagrangian $L$ with $L$ (and $e^{iL}$) invariant under $G$, and the problem
is whether it can be modified to $L'$ so that $e^{iL'}$ is invariant under the
twisted action. The powers of $L$ span a coalgebra all of whose elements are
$G$-invariant. Conversely, given a coalgebra $C$ all of whose elements are
invariant under some group action, there is a canonical $G$-invariant
group-like element associated to this coalgebra with coefficients in the dual
algebra of $C$. So a fundamental question is whether the maximal coalgebra in
the space of $G$-invariant classical actions is isomorphic to the maximal
coalgebra in the space of actions invariant under the twisted action of $G$.

The simplest case is when one can find a $G$-invariant Feynman measure, in
which case the cocycle is trivial and the twisted action of $G$ is the same as
the untwisted action. In terms of the cocycle \ above, $\rho \omega$ is
invariant for some renormalization $\omega$ if and only if $\rho_g = \rho^{-
1} g (\rho)$ for all $g$ (where $g (\rho) = g \rho g^{- 1}$), in other words
there is an invariant measure $\omega$ if and only if the cocycle is a
coboundary. \ This case happens, for example, when spacetime $M$ is Minkowski
space and $G$ is the Lorentz or Poincare group (or one of their double
covers). Dimensional regularization in this case is automatically
$G$-invariant, and so gives a $G$-invariant Feynman measure.

In the case of BRST operators, there need not be any $G$-invariant Feynman
measure. In this case the following theorem shows that one can find suitable
coalgebras provided that certain obstructions, called anomalies, all vanish.
The renormalizations $\rho_g$ need not preserve products in $S \Gamma \omega
SJ \Phi$, but do preserve the coproduct and also fix all elements of
$\Gamma \omega SJ \Phi$ if they are normalized as in example
\ref{simple operator}. So we have an action of $G$ on the space $V = \Gamma
\omega SJ \Phi$, which lifts to two different actions of the coalgebra
$\tmop{SV}$, the first $\sigma_1 (g)$ preserving the product, and the second
$\sigma_2 (g) = \rho_g \sigma_1 (g)$ given by twisting the first by the
cocycle $\rho_g$.

\begin{theorem}
  Suppose that $V$ is a real vector space acted on by a group $G$, and there
  are two extensions $\sigma_1$. $\sigma_2$ of this action to the coalgebra
  $SV$. \ If the cohomology group $H^1 (G, V)$ vanishes then the maximal
  coalgebras in $SV^{}$ whose elements are fixed by these 2 actions of $G$ are
  isomorphic under an isomorphism fixing the elements of $V$. 
\end{theorem}

\begin{proof}
  We construct an isomorphism $f$ from the maximal coalgebra in the space \ of
  $\sigma_1$-invariant elements to the maximal coalgebra in the space of
  $\sigma_2$-invariant elements by induction on the degree of elements. We
  start by taking $f$ to be the identity map on elements of degree at most 1.
  We can assume that the 2 actions coincide on elements of degree less than
  $n$, and have to find an isomorphism $f$ making them the same on elements of
  degree $n$, which we will do by adding elements of $V$ to a basis of the
  elements of degree $n$. Suppose that $a$ is an element of degree $n > 1$
  contained in a coalgebra of $G$-invariant elements. We want to find $v \in
  V$ so that
  \[ \sigma_1 (g) (a + v) = \sigma_2 (g) (a) + v \]
  or equivalently
  \[ \sigma_1 (g) (v) - v = \sigma_2 (g) (a) - a. \]
  The right hand side, as a function of $g$, is a 1-coboundary of an element
  $a \in SV$, and therefore a 1-cocycle. We show that the right hand side is
  in $V$. We have
  \[ \Delta (a) = a \otimes 1 + 1 \otimes a + \sum_i b_i \otimes c_i \]
  for some \ elements $b_i$ and $c_i$ of degrees less than $n$ invariant under
  $G$ (for both actions, which coincide on elements of degree less than $n$).
  Applying \ $\sigma_2$ we find that $\Delta (\sigma_2 (g) a) = \sigma_2 (g) a
  \otimes 1 + 1 \otimes \sigma_2 (g) a + \sum_i b_i \otimes c_i$, so
  subtracting these two identities shows that $\sigma_2 (g) (a) - a$ is a
  primitive element of $\tmop{SV}$ and therefore in $V$. Therefore the right
  hand side, as a function of $g$, is a 1-cocycle with values in $V$. \ The
  solvability of the condition for $v$ says exactly that this expression is
  the \ coboundary of some element $v \in V$. In other words the obstruction
  to finding a suitable $v$ is exactly an element of the cohomology group $H^1
  (G, V)$, so as we assume this group vanishes we can always solve for $v$. 
\end{proof}

\begin{example}
  We take $V$ to be $\Gamma \omega SJ \Phi$, and $G$ to be some group
  acting on $V$. Then the spaces of classical and quantum actions are
  coalgebras acted on by $G$, whose primitive elements can be identified with
  $V$. If $H^1 (G, \Gamma \omega SJ \Phi)$ vanishes, then the maximal
  $G$-invariant coalgebra in the coalgebra of classical actions is isomorphic
  to the maximal $G$-invariant coalgebra in the coalgebra of quantum actions.
  So if $L$ is a $G$-invariant classical Lagrangian, then $e^L$ is a
  $G$-invariant classical action, so gives a $G$-invariant quantum action. One
  cannot get a $G$-invariant quantum action by exponentiating a $G$-invariant
  quantum Lagrangian because the space of quantum actions does not in general
  have a $G$-invariant product. 
\end{example}

\begin{example}
  Sometimes the group $G$ only fixes classical Lagrangians up to boundary
  terms, in other words the Lagrangian is a $G$-invariant element of $\Gamma
  \omega SJ \Phi / D$. In this case one replaces the cohomology group
  $H^1 (G, \Gamma \omega SJ \Phi)$ by $H^1 (G, \Gamma \omega SJ
  \Phi / D)$.
\end{example}

The element $e^{iL_F}$ lies in the completion of $S \Gamma \omega SJ
\Phi$ and is fixed by the zeroth order part of the BRST operator. So the BRST
operator acts on $e^{iL_F} S \Gamma \omega SJ \Phi$.

The groups $H^1 (G, \Gamma \omega SJ \Phi) \tmop{and} H^1 (G, \Gamma
\omega SJ \Phi / D)$ (and their variations for Poincare invariant
Lagrangians) for the BRST operators of gauge theories have been calculated in
many cases, at least for the case of Minkowski space \ (see for example
Barnich, Brandt, and Henneaux {\cite{Barnich}}) and are sometimes zero, in
which case \ corresponding invariant \ quantum field theories exist.

\end{document}